\documentclass[letterpaper, 10 pt, conference]{ieeeconf}  

\IEEEoverridecommandlockouts                              

\usepackage{amssymb}  
\usepackage{xcolor}
\usepackage{graphicx}
\usepackage{bbm}

\usepackage{mathtools}
\DeclarePairedDelimiter\floor{\lfloor}{\rfloor}
\DeclarePairedDelimiter\ceil{\lceil}{\rceil}

 \newtheorem{lem}{Lemma}[section]
 \newtheorem{thm}{Theorem}[section]
 \newtheorem{coro}{Corollary}[section]
 
 \newtheorem{defn}{Definition}[section]
 \newtheorem{rem}{Remark}[section]

\title{\LARGE \bf
Finite-sample analysis of identification of switched linear systems with arbitrary or restricted switching*
}

\author{Shengling Shi$^{1}$, Othmane Mazhar$^{2}$ and Bart {De Schutter}$^{1}$
\thanks{*This research has received funding from the
European Research Council (ERC) under the European Union’s Horizon
2020 research and innovation programme (Grant agreement No.101018826 - CLariNet).\newline \indent   \copyright  2022 IEEE. Personal use of this material is permitted.  Permission from IEEE must be obtained for all other uses, in any current or future media, including reprinting/republishing this material for advertising or promotional purposes, creating new collective works, for resale or redistribution to servers or lists, or reuse of any copyrighted component of this work in other works.}
\thanks{$^{1}$Shengling Shi and Bart De Schutter are with the Delft Center for Systems and Control, Delft University of Technology, the Netherlands
        {\tt\small \{s.shi-3, b.deschutter\}@tudelft.nl}}%
\thanks{$^{2}$Othmane Mazhar is with the Department of Mathematics, KTH Royal Institute of Technology, Sweden
        {\tt\small othmane@kth.se}}%
 }

\begin{document}

\maketitle
\thispagestyle{empty}
\pagestyle{empty}

\begin{abstract}
For the identification of switched systems with a measured switching signal, this work aims to analyze the effect of switching strategies on the estimation error. The data for identification is assumed to be collected from globally asymptotically or marginally stable switched systems under switches that are arbitrary or subject to an average dwell time constraint. Then the switched system is estimated by the least-squares (LS) estimator. To capture the effect of the parameters of the switching strategies on the LS estimation error, finite-sample error bounds are developed in this work. The obtained error bounds show that the estimation error is logarithmic of the switching parameters when there are only stable modes; however, when there are unstable modes, the estimation error bound can increase linearly as the switching parameter changes. This suggests that in the presence of unstable modes, the switching strategy should be properly designed to avoid the significant increase of the estimation error.
\end{abstract}

\section{INTRODUCTION}
The finite-sample error analysis of identification methods has recently received considerable attention \cite{matni2019self}. When the estimated model is used for controller design, the obtained error bound is important in understanding the effect of the estimation error on the control performance \cite{matni2019self}. While several works consider the finite-sample analysis of linear system identification \cite{simchowitz2018learning,sarkar2019near,jedra2020finite,
djehiche2021non,sarkar2021finite,oymak2021revisiting}, the finite-sample analysis for identifying hybrid systems has rarely been addressed \cite{lauer2019hybrid}. In this work, we consider the identification of a particular class of hybrid systems, i.e., switched linear systems (which consist of multiple linear systems corresponding to different modes). 

A practical setting is when the switching signal of a switched system is measured, e.g., servo turntable systems \cite{zhang2016switched} and power systems with the switching signal as an input for stabilization \cite{chen2014application}. In this setting, the statistical analysis of the identification problem has been considered in \cite{sarkar2019nonparametric,sattar2021identification,sayedana2022switch}. The authors of \cite{sarkar2019nonparametric} consider a more general setup where the outputs are measured instead of the states, and a subspace identification method is employed. However, the estimation procedure requires collecting data from multiple independent trajectories obtained by restarting the system. When both the states and the switching signal are measured, the statistical analysis of the switched LS estimator, i.e., applying the standard LS estimator for every mode separately, has been addressed recently in \cite{sattar2021identification,sayedana2022switch}. The extension of the analysis from the standard LS to switched LS is non-trivial, as the covariances of the local estimators are coupled through the system dynamics \cite{sarkar2021finite}. 

The authors in \cite{sayedana2022switch} have established the
consistency of the switched LS estimator; however, the result is asymptotic and thus valid only when the data length approaches infinity. The work \cite{sattar2021identification} addresses the finite-sample analysis of the switched LS estimator. The employed estimator requires sub-sampling the data and knowing the noise covariance, typically an unknown quantity. Furthermore, all the above three works model the switching signal as a stochastic process, i.e., an i.i.d. process or a Markov chain. This model may not be suitable for some situations, e.g., when the switching is an external input or caused by state-space partition.

This work aims to derive a finite-sample estimation error bound for the LS estimation of switched linear systems from the measured states and the measured switching signal. We consider the typical classes of deterministic switching signals, including arbitrary switching and switching with an average dwell time constraint \cite{lin2009stability}. Under the considered classes of switching signals, we assume the nominal switched system to be globally marginally or asymptotically stable.

A preliminary estimation error bound is first obtained by extending the results in \cite{sarkar2019near} for linear systems; however, the resulting bound contains the Gramian of the switched system, which depends on the measured switching sequence. To obtain an error bound that generally holds, independently of any realization of the switching signal, and that captures the effect of the parameters of the switching strategies, data-independent bounds on the spectrum of the Gramian are developed and then combined with the preliminary bound. In summary, the contributions of this work are as follows:
\begin{itemize}
    \item The existing data-independent finite-sample error bound in \cite{sarkar2019near} for linear system identification is extended to switched systems.
    
    \item Data-independent bounds for the spectrum of the Gramian are developed for the switched system.
\end{itemize}

\subsection*{Notation}
Given any positive integer $k$, $[k]$ denotes the set $\{1,2,\dots,k\}$. Given any real matrix $M$, $\rho(M)$ denotes its spectral radius, $\sigma_{\mathrm{max}}(M)$ denotes its maximum singular value, $\sigma_{\mathrm{min}}(M)$ denotes the minimum singular value, and $M^\dagger$ denotes its pseudoinverse. Given two symmetric matrices $M$ and $H$, $M \leqslant H$ and $M< H$ means that $H- M$ is positive semi-definite and positive definite, respectively. Given any real number $c$, $\floor{c}$ and $\ceil{c}$ denote the floor function and the ceiling function, respectively. For the summation $\sum_{j=k_0}^{k_1} a_j$ of any sequence with some integers $k_0$, $k_1$, we define $\sum_{j=k_0}^{k_1} a_j \triangleq 0$ if $k_1<k_0$. 

\section{Problem formulation}
Consider the discrete-time switched linear system:
\begin{align}
    x_{t+1} & = A_{w_t} x_t + e_t,  \label{eq:system}
\end{align}
where $t \in \mathbb{Z}^+$ is the time step, $x_t \in \mathbb{R}^n$ is the state vector, $w_t$ is the deterministic switching signal satisfying $w_t \in [s]$ with $s$ a positive integer, and $e_t$ is a sub-Gaussian distributed white noise vector with variance proxy $1$ and $\mathbb{E}(e_t e_t^\top) =I$. The (nominal) system is said to be (globally) marginally stable\footnote{While different definitions for marginal stability exist, we follow the notion in \cite{morris2021marginally}.} if, when the noise is absent, there exists a $b$ such that $\|x_t\|_2 \leqslant b \|x_0\|_2$ for any $t$ and any $x_0$. Similarly, (nominal) global asymptotic stability is defined when the noise is absent.

We consider the LS estimation of the switched system given the measurements in $\{(x_t, w_t)\}_{t=1}^N$, when the switched system is (globally) asymptotically or marginally stable. For simplicity, we assume $x_0=0$ for the data collection. For any mode $i \in [s]$, let $\mathcal{T}_i \subseteq [N-1]$ denote the subset of time steps when mode $i$ is active, i.e., $w_{t} = i$ for all $t \in \mathcal{T}_i $, and $N_i \in \mathcal{T}_i$ denotes the last time step when mode $i$ is active.

Then the  LS estimator for mode $i$ is 
\begin{align}
\hat{A}_i &= \arg\min_{A_i} \sum_{t \in \mathcal{T}_i} \|x_{t+1} - A_{i} x_{t}\|_2^2, \nonumber \\
& = \Big(  \sum_{t \in \mathcal{T}_i} x_{t+1} x_{t}^\top   \Big) \Big( \sum_{t \in \mathcal{T}_i} x_{t}  x_{t}^\top \Big)^\dagger. \label{eq:LS}
\end{align}
Equation~\eqref{eq:LS} leads to \begin{equation} \label{eq:LSerror}
\hat{A}_i - A_i = S_i X_i^\dagger,    
\end{equation}
where
$
S_i \triangleq \sum_{t \in \mathcal{T}_i} e_{t} x_{t}^\top, \text{ }  X_i \triangleq    \sum_{t \in \mathcal{T}_i} x_{t} x_{t}^\top.
$
Therefore, the main goal is to derive a high-probability error bound for $\|S_i X_i^\dagger\|_2$. In addition, we focus on developing an error bound that is data-independent and captures the dependence on the parameters of the switching strategies. Data-independent bounds are more theoretically informative than data-dependent ones as they reveal how they scale with the properties of the unknown system and the parameters of the switching signal \cite{matni2019self}. They also provide the worst-case guarantees as they hold for any realization of the data.

Before we address the above problem, let us define some notations. For any time step $t$, we adopt a shorthand notation:
$
A_{(t)} \triangleq A_{w_t}.
$
For any two positive integers $j \geqslant k$, we define
$
A_{(j:k)} \triangleq A_{(j)} A_{(j-1)} \dots A_{(k)}
$
and $A_{(j:k)} \triangleq I$ when $j<k$. Then for any $t$, we have 
\begin{equation}
x_t = A_{(t-1:0)} x_0+ \sum_{i=0}^{t-1} A_{(t-1:t-i)}e_{t-1-i}. \label{eq:StateEvolution}
\end{equation}
We define the Gramian of the system:
\begin{equation}
\Gamma_t \triangleq \sum_{i=0}^{t-1} A_{(t-1:t-i)} A_{(t-1:t-i)}^\top, \label{eq:GammaDef}
\end{equation}
and it can be found that $\mathbb{E}(x_t x_t^\top) = \Gamma_t$. 
\begin{rem}
The results in this work remain valid if $e_t$ is replaced by a more general noise source $\eta_t = \sigma_e e_t$ in \eqref{eq:system}, which has $\mathbb{E}(e_t e_t^\top) =\sigma_e^2 I $. This new noise will lead to the same LS estimate \eqref{eq:LS} due to the cancellation of $\sigma_e$ in \eqref{eq:LS}.
\end{rem}

\section{Preliminary error bound}
In this section, we obtain a preliminary high-probability bound for $\|A_i - \hat{A}_i\|_2$ by applying the result for linear system identification. In particular, we start from the finite-sample bound in \cite{sarkar2019near}. While there are other bounds available for linear systems, the one in \cite{jedra2020finite} is derived for asymptotic stable system only, and the one in \cite{simchowitz2018learning} has an additional parameter introduced by the analytical method.

Since the switching signal is deterministic, the bound in \cite{sarkar2019near} for linear systems extend to the estimation error of this work with the difference that the Gramian for linear systems is replaced by the Gramian in \eqref{eq:GammaDef}. To this end, we first define
$
T^{(0)} \triangleq  K \big(n+ \ln \frac{2}{\delta} \big), 
$
where $K$ is some positive constant. Then following \cite[Sections 4, 9]{sarkar2019near} analogously, the following high-probability error bound for one mode can be obtained.

\begin{thm} \label{thm:MainThm}
For any $\delta \in (0,1/4)$, if $n \geqslant 2$ and $|\mathcal{T}_i | \geqslant $
\begin{equation} \label{eq:Inequality}
\max\{ T^{(0)}, \text{ } 64n \ln\big(   \mathrm{tr}( \Gamma_{ N_i} -I ) + 1\big) + 128n \ln \Big(  \frac{5}{ \delta}\Big) \},
\end{equation}
then with probability at least $1-4\delta$, we have $\|\hat{A}_i - A_i \|_2 \leqslant $
\begin{equation} \label{eq:PreBound2}
\frac{1}{\sqrt{|\mathcal{T}_i |}} \sqrt{32n \Big[\frac{1}{2} \ln\big(4 \mathrm{tr}(\Gamma_{N_i})+1 \big) + \ln \Big( \frac{5}{ \delta }\Big) \Big] }.
\end{equation}
\end{thm}
The bound in \eqref{eq:PreBound2} decreases as $|\mathcal{T}_i |$ increases, and it holds uniformly for all the modes with probability at least $1-4s\delta$ for any $\delta \in (0,1/(4s))$. The lower bound \eqref{eq:Inequality} requires a mode  to be visited sufficiently often and is related to the persistent excitation of the state measurements, i.e., it ensures that $\sum_{t \in \mathcal{T}_i} x_{t}  x_{t}^\top$ in \eqref{eq:LS} is invertible and also well-conditioned with high probability. Similar requirements for the switching sequence also appear, e.g., in \cite{vidal2008recursive}.
\begin{rem}
The constant $K$ in $T^{(0)}$ is due to a concentration inequality of sub-Gaussian random matrices, see \cite[Propsition 8.3]{sarkar2019near}. Concentration inequalities and the resulting finite-sample results are typically not precise and hold up to some unspecified constants, as discussed in \cite{vershynin2010introduction} and shown in the results of \cite{simchowitz2018learning,sarkar2019near,jedra2020finite}. The main objectives of interest are typically the change rate of the guarantees when important parameters change, e.g., the sample size or the state dimension.\hfill  $\blacksquare$
\end{rem}

For the general case where $n \geqslant 1$, the bound \eqref{eq:PreBound2} admits a less compact formulation:
\begin{coro}
For any $\delta \in (0,1/4)$, if $|\mathcal{T}_i | \geqslant $
\begin{equation*}
\max\{ T^{(0)}, \text{ }  64n \ln\big(  \mathrm{tr}(\Gamma_{N_i} -I   ) + 1\big) + 128 \ln \Big(  \frac{5^n}{ \delta^{1+n/2}}\Big) \}, 
\end{equation*}
then with at least probability $1- 4\delta$, we have $\|\hat{A}_i - A_i \|_2  \leqslant $
\begin{equation*} 
\frac{1}{\sqrt{|\mathcal{T}_i |}} \sqrt{32n \Big[ \frac{1}{2}\ln\big(4 \mathrm{tr}(\Gamma_{N_i} )+1 \big) + \ln \Big( \frac{5}{ \delta^{1/n+1/2}}\Big) \Big] }.
\end{equation*}
\end{coro}

The bound \eqref{eq:PreBound2} depends on the switching sequence $\{w_t\}_{t=1}^N$ due to the Gramian $\Gamma_{N_i}$.  Given that $\mathrm{tr}(\Gamma_T) \leqslant n \lambda_{\mathrm{max}}(\Gamma_{T})$, we will further upper bound $\lambda_{\mathrm{max}}(\Gamma_{T})$ for a certain class of switching signals using the properties of the switching class. Combining this upper bound with the results in this section will lead to estimation error bounds that capture the dependence on the parameters of switching strategies. These error bounds can reveal how the properties of the switching signal influence the estimation error.

Note that the derived bounds for the Gramian in the following sections are also applicable to other finite-sample bounds for linear system identification \cite{simchowitz2018learning,jedra2020finite,djehiche2021non} when extended to switched systems, as the Gramian is an essential object in these bounds. In addition, since $\lambda_{\mathrm{min}}(\Gamma_T)$ can also be of interest, e.g., in \cite{simchowitz2018learning,jedra2020finite}, the analysis of the lower bound for $\lambda_{\mathrm{min}}(\Gamma_T)$ will also be exploited in this work.

\section{Spectral properties of the Gramian}
To derive a  bound for $ \lambda_{\mathrm{max}}(\Gamma_{T})$ that captures the parameters of the switching strategies, the properties of the switching signal should be further specified. In this work, we consider the typical classes of deterministic switching signals that are arbitrary or under time restriction \cite{lin2009stability}.

In addition, both sides of \eqref{eq:Inequality} depend on $|\mathcal{T}_i |$; this is clear when only a single mode is active, i.e., $N_i = |\mathcal{T}_i |$. Thus, $ \lambda_{\mathrm{max}}(\Gamma_{T})$ should not grow too fast as $T$ increases; otherwise, there may not exist a $\mathcal{T}_i$ for \eqref{eq:Inequality} to hold. To control the growth rate of $ \lambda_{\mathrm{max}}(\Gamma_{T})$, the (nominal) switched system is assumed to be asymptotically or marginally stable under switching.

\subsection{Arbitrary switching}
We first consider systems that are marginally or asymptotically stable under arbitrary switching, i.e., any switching sequence. The following stability condition follows immediately from \cite[Theorem 6]{lin2009stability}.
\begin{lem}\label{lem:StableArbitrarySwitch}
The switched system is globally marginally stable under arbitrary switching (or asymptotically stable) if there exists a positive integer $m$ such that $\|A_{s_1}\dots A_{s_m}\|_2 \leqslant 1$ (or $\|A_{s_1}\dots A_{s_m}\|_2 < 1$) for all $s_j \in [s]$ and $j=1,\dots,m$.
\end{lem}

In the above case, we say that the switched system is globally marginally or asymptotically stable with stability horizon $m$. Here, $m$ can be interpreted as a safe time horizon, within which any switching sequence will not affect stability. It has been shown in \cite{lin2009stability} that if global asymptotic stability and the $\infty$-norm are considered instead, the above condition is sufficient and necessary. In this work, we consider the $2$-norm to facilitate our analysis of $\lambda_{\mathrm{max}}(\Gamma_T)$. With this result, we aim to upper bound $\lambda_{\mathrm{max}}(\Gamma_T)$ as a function of $m$.
\begin{thm} \label{thm:ArbitrarySwitchMarginal}
  Define $\sigma_{\mathrm{max}} \triangleq \max_{i \in [s]} \sigma_{\mathrm{max}}(A_i)$ and $\sigma_{\mathrm{min}} \triangleq \min_{i \in [s]} \sigma_{\mathrm{min}}(A_i)$. If there exist a positive integer $m$ and two real numbers $a_{\mathrm{min}}, a_{\mathrm{max}} \in [0,1]$ such that $a_{\mathrm{min}} \leqslant \sigma_{\mathrm{min}}(A_{s_1}\dots A_{s_m}) \leqslant \|A_{s_1}\dots A_{s_m}\|_2 \leqslant a_{\mathrm{max}}$ for all $s_j \in [s]$ and $j=1,\cdots,m$, then it holds that
\begin{align}
& \Big( \sum_{i=0}^{m-1} \sigma_{\mathrm{min}}^{2i} \Big) \sum_{j=0}^{\floor{(T-1)/m}}a_{\mathrm{min}}^{2j}  
\leqslant
\lambda_\mathrm{min}(\Gamma_T) \nonumber \\& \leqslant \lambda_\mathrm{max}(\Gamma_T) \leqslant   \Big( \sum_{i=0}^ {m-1} \sigma_{\mathrm{max}}^{2i} \Big) \sum_{j=0}^{\floor{(T-1)/m}}a_{\mathrm{max}}^{2j}. \label{eq:BoundArbi}
\end{align}
\end{thm}
\begin{proof}
For any $i \in \{0,1,\dots, T-1\}$, it holds that $i = \floor{i/m} m+b_i$ for some non-negative integer $b_i < m$. This leads to $\|A_{(T-1,T-i)}\|_2 \leqslant$
\begin{align}
& \|A_{(T-1,T-b_i)} \|_2 \|A_{T-b_i-1,T-b_i-m}\|_2 \dots  \| A_{(T-i+m-1,T-i)}\|_2 \nonumber \\
&\leqslant  a_{\mathrm{max}}^{\floor{i/m}} \|A_{(T-1,T-b_i)} \|_2   \leqslant \sigma_{\mathrm{max}}^{b_i} a_{\mathrm{max}}^{\floor{i/m}}. \label{eq:ProofMinDwell1}
\end{align}
Since $\floor{i/m}\in \{0,\dots, \floor{(T-1)/m}\}$ and $b_i<m$, it holds that $\lambda_\mathrm{max}(\Gamma_T) \leqslant$
\begin{align*}
& \sum_{i=0}^{T-1} \| A_{(T-1:T-i)}\|_2^2 \nonumber \leqslant
 \sum_{j=0}^{\floor{(T-1)/m}} \Big[ a_{\mathrm{max}}^{2j} \Big( \sum_{b=0}^ {m-1} \sigma_{\mathrm{max}}^{2b} \Big) \Big], 
\end{align*}
which concludes the upper bound. The lower bound follows similarly from the fact that 
$
 \lambda_\mathrm{min}(\Gamma_T) \geqslant \sum_{i=0}^{T-1} \sigma_{\mathrm{min}}(A_{(T-1:T-i)})^2.
$
\end{proof}
The upper bound can be simplified in special cases.
\begin{coro} \label{coro:Arbitra}
In the setting of Theorem~\ref{thm:ArbitrarySwitchMarginal}, if $\sigma_{\mathrm{max}} \not= 1$ holds additionally, we have \begin{equation}\label{eq:BoundArbiCoro}
    \lambda_\mathrm{max}(\Gamma_T) \leqslant p (\floor{T/m}+1), \end{equation} where $p \triangleq (1- \sigma_{\mathrm{max}}^{2m})/(1-\sigma_{\mathrm{max}}^2 ) $; \hfill  $\blacksquare$
\end{coro}
The above corollary also covers the situation with $\sigma_{\mathrm{max}} >1$, and combining it with \eqref{eq:PreBound2} can lead us to an estimation error bound that depends on the stability horizon $m$:
\begin{coro} \label{coro:ArbitrarySwitch}
For any $\delta \in (0,1/4)$, if the switched system is globally marginally stable under an arbitrary switching signal with stability horizon $m$, and if it holds that $\sigma_{\mathrm{max}} \not= 1$, $n \geqslant 2$ and $|\mathcal{T}_i | \geqslant \max\{ T^{(0)}, $
\begin{equation}
    \text{ }  64n \ln\Big(n  \big[p (\floor{N_i/m}+1)-1\big] + 1\Big) + 128n \ln \Big(  \frac{5}{ \delta}\Big) \}, \label{eq:lowerBoundT}
\end{equation}
then with probability at least $1-4\delta$, we have $\|\hat{A}_i - A_i \|_2 \leqslant $
\begin{equation*} 
\frac{1}{\sqrt{|\mathcal{T}_i |}} \sqrt{32n \Big[ \frac{1}{2}\ln \big[4 np (\floor{N_i/m}+1)+1 \big] + \ln \Big( \frac{5}{ \delta }\Big) \Big] }.
\end{equation*}
\end{coro}

The above error bound is logarithmic of $1/m$, and thus the increase of the stability horizon $m$ leads to a slow decrease in the error bound. Intuitively, given a data length $N$, a larger $m$ leads to less informative state measurements and thus a smaller Gramian in \eqref{eq:PreBound2}, which decreases the error bound. The decay rate of the bound in terms of the data length is $\mathcal{O}\big( \sqrt{(\ln N) / |\mathcal{T}_i | } \big)$, which agrees with the asymptotic analysis in \cite[Corrolary~4]{sayedana2022switch}. Furthermore, when the nominal system is asymptotically stable, i.e., $a_\mathrm{max}<1$, we have $\lambda_{\mathrm{max}}(\lambda_T) \leqslant m/(1-a^2_{\mathrm{max}})$, and combing this bound with \eqref{eq:PreBound2} shows that the estimation error of each mode is $\mathcal{O}\big( 1 /\sqrt{ |\mathcal{T}_i | }\big)$, which matches the optimal decay rate of the LS estimator for asymptotically stable linear systems \cite{jedra2020finite}. In addition, \eqref{eq:lowerBoundT} requires $|\mathcal{T}_i|$ to scale with the state dimension as $ \mathcal{O}\big( n\ln(n)\big)$, which is in line with the rate in \cite{simchowitz2018learning}.

\begin{rem} The bound in Corollary~\ref{coro:ArbitrarySwitch} is pseudo-data-independent, as $|\mathcal{T}_i|$ still varies over different switching sequences in the considered class of switching signals, i.e., arbitrary switching in this section. We choose to keep $|\mathcal{T}_i|$ in the bound, as the data size of a particular mode is valuable information for identifying one mode. In some situations, e.g., with a stochastic switching that visits every mode with a positive probability, it is possible to substitute $|\mathcal{T}_i|$ by a function of the total sample size $N$, e.g., see a similar step taken in \cite{sayedana2022switch} for the asymptotic analysis. Then the bound becomes less tight but completely data-independent. 
\end{rem}

\subsection{Minimum dwell time}
Time-restricted switching is a standard switching strategy in the control of switched systems \cite{lin2009stability}. The intuition is that if the system does not switch too often or stay too long at unstable modes, the overall switched system can be stable. We first consider switching with a dwell time constraint, where each mode is stable, and the system stays in each mode for a sufficiently long time such that the overall system is stable. The concept of dwell time is defined as follows.
\begin{defn}(\cite{lin2009stability})
A positive integer $\tau$ is called a dwell time of a switching signal if the time interval between two consecutive switchings is not smaller than $\tau$.
\end{defn}

To characterize $\lambda_{\mathrm{max}}(\Gamma_T)$ using the dwell time, we first define several new variables that capture the system properties. Let all $A_i$ be Schur stable in this subsection, and there always exist real constants $\rho <1$ and $c_i \geqslant 1$ such that $\|A_i^k\|_2 \leqslant c_i \rho^k$ holds for any positive integer $k$ and any $i \in [s]$ \cite{zhai2002qualitative}. Then, we define 
$
c \triangleq \max_{i \in [s]} c_i.
$

Finally, when the dwell time of the switching signal is larger than a minimum dwell time $\tau^\star$ such that the switched system is marginally stable or asymptotically stable, we can upper bound $\lambda_{\mathrm{max}}(\Gamma_T)$ as a function of the minimum dwell time and the constants $c$, $\rho$.

\begin{thm} \label{thm:MiniDwell} 
Suppose that $\rho(A_i) < 1$ holds for all $i \in [s]$, and
let $\tau^\star$ be any positive integer such that $\|A_i^{\tau^\star} \|_2 \leqslant a \triangleq c \rho^{\tau^\star} \leqslant 1$. If the dwell time $\tau$ of the switching signal in \eqref{eq:system} satisfies $\tau \geqslant \tau^\star$, we have 
\begin{itemize}
    \item 
    \begin{equation} \label{eq:DewellTimeBound}
\lambda_{\mathrm{max}}(\Gamma_T) \leqslant 1+ c^4 \frac{\rho^2}{1-\rho^2} \Big(1+ \frac{T}{ \tau^\star}
     \Big); \end{equation}
     
     \item if $a <1$ also holds, we have
     \begin{equation} \label{eq:DewellTimeBound2}
         \lambda_{\mathrm{max}}(\Gamma_T) \leqslant 1+ c^4 \frac{\rho^4}{1-\rho^2} \Big(1+  \frac{1-a^{\floor{T/\tau^\star}}}{1-a} \Big).
     \end{equation}
\end{itemize}
\end{thm}
\begin{proof}
Let $P$ denote the number of switches within the time step interval $[0,T-1]$, and for any $j \in [P] $, $t_j$ denotes the first time step of the new mode after the $j$-th last switch, e.g., $t_1$ is the first time step after the last switch.  Therefore, for the term $A_{(T-1,T-i)}$ in $\Gamma_T$ and if $T-i < t_1$, we have for some $j \in [P]$, $\|A_{(T-1,T-i)}\|_2  \leqslant$
\begin{align*}
& \|A_{(T-1,t_1)}\|_2 \dots \| A_{(t_{j-1}-1,t_{j})} \|_2 \| A_{(t_j-1,T-i)} \|_2 \\
& \leqslant a^{j-1} \|A_{(T-1,t_1)}\|_2 \| A_{(t_{j}-1,T-i)} \|_2 ,
\end{align*}
where the last inequality follows from the minimum dwell time condition, and $a = c\rho^{\tau^\star} \leqslant 1$.

According to the defined variables $c$ and $\rho$, it holds that
$$
\| A_{(t_{j}-1,T-i)} \|_2  \leqslant c \rho^{t_{j}-T+i  }, \text{ }  \|A_{(T-1,t_1)}\|_2  \leqslant c \rho^{T-t_1}.
$$
If $T-i \geqslant t_1$,
$
\|A_{(T-1,T-i)}\|_2 \leqslant c \rho^{i}.
$
Therefore, we have
\begin{align*}
& \lambda_{\mathrm{max}}(\Gamma_T) \leqslant  \sum_{i=0}^{T-1} \|A_{(T-1,T-i)}\|_2^2  \leqslant 1+ \sum_{j=1}^{T-t_1} (c \rho^j) ^2 +  \\ &    (c \rho^{T-t_1})^2 \Big[ \sum_{j=1}^{P-1} a^{j-1} \Big(\sum_{i=1}^{t_{j}-t_{j+1}} (c \rho^{i})^2 \Big) 
 +  a^{P-1} \sum_{k=1}^{t_{P}-1} (c \rho^k)^2 \Big]. 
\end{align*}
Given $\rho<1$ and $a \leqslant 1$, for any positive integer $N$, we have $\sum_{k=1}^N \rho^{2k} \leqslant \frac{\rho^2}{1-\rho^2}$, which leads to
\begin{align*}
   & \lambda_{\mathrm{max}}(\Gamma_T)  \leqslant 1+  c^2 \frac{\rho^2}{1-\rho^2} \Big(1+ P c^2 \rho^{2(T-t_1)} 
     \Big) \\
     & \leqslant  1+ c^2 \frac{\rho^2}{1-\rho^2} \Big(1+ P c^2  
     \Big)   \leqslant 1+ c^4 \frac{\rho^2}{1-\rho^2} \Big(1+ \frac{T}{ \tau^\star}
     \Big). 
\end{align*}
Finally, if $a<1$, we have 
$
    \lambda_{\mathrm{max}}(\Gamma_T)  \leqslant  1+ c^4 \frac{\rho^2}{1-\rho^2} (1+ \sum_{j=1}^P a^{j-1}
     ). 
$
The fact that $P \leqslant \floor{T/\tau^{\star}}$ together with the above equation concludes the last bound.
\end{proof}
The first bound \eqref{eq:DewellTimeBound} is valid when the switched system is marginally stable, which is guaranteed by $\|A_i^{\tau^\star} \|_2 \leqslant c \rho^{\tau^\star} \leqslant 1$. The bound \eqref{eq:DewellTimeBound} shows that a smaller $\tau^\star$ leads to a larger bound for $\lambda_{\mathrm{max}}(\Gamma_T)$, which can be interpreted as the effect of more frequent switching on the more informative states. The bound \eqref{eq:DewellTimeBound2} is valid when the switched system is globally asymptotically stable. Then combining \eqref{eq:DewellTimeBound} and \eqref{eq:PreBound2} leads to the following estimation error bound.
\begin{coro} \label{coro:MinimumDwell}
In the setting of Theorem~\ref{thm:MiniDwell}, for any $\delta \in (0,1/4)$, if it holds additionally that $n \geqslant 2$ and $|\mathcal{T}_i | \geqslant \max\{ T^{(0)},$
$$
\text{ }  64n \ln\big[n \frac{c^4 \rho^2  }{1-\rho^2}\Big( 1+ \frac{N_i}{\tau^\star}\Big) + 1\big] + 128n \ln \Big(  \frac{5}{ \delta}\Big) \},
$$
then with probability at least $1-4\delta$, we have $\|\hat{A}_i - A_i \|_2 \leqslant $
\begin{equation*} 
\frac{1}{\sqrt{|\mathcal{T}_i |}} \sqrt{32n \Big[\ln \frac{1}{2}\big[4 n L +1 \big] +\ln \Big( \frac{5}{ \delta }\Big) \Big] },
\end{equation*}
where $L=\big( 1+ c^4 \frac{\rho^2}{1-\rho^2} (1+ \frac{N_i}{ \tau^\star}) \big) $.
\end{coro}

Given $|\mathcal{T}_i |$, the above bound is logarithmic of $1/\tau^\star$, and thus a smaller minimum dwell time leads to a slow increase of the error bound, while the error bound is dominated by the sample size $|\mathcal{T}_i |$ and the state dimension $n$.
\begin{rem} \label{rem:LowerBound}
Unlike Theorem~\ref{thm:ArbitrarySwitchMarginal}, a lower bound for $\lambda_{\mathrm{min}}(\Gamma_T)$, which has a similar form as the upper bound, is not derived in Theorem~\ref{thm:MiniDwell}, due to the complexity of the analysis. In this case, the following straightforward lower bounds can be considered: $\lambda_{\mathrm{min}}(\Gamma_T) \geqslant 1$ or a tighter one, $\lambda_{\mathrm{min}}(\Gamma_T) \geqslant 1+ \sum_{j=1}^{T-1}\sigma_{\mathrm{min}}^{2j}$. \hfill $\blacksquare$
\end{rem}
\section{Spectral properties of the Gramian with average dwell time}
The so-called average dwell time constraint limits the number of switches in each time period, and it is less restrictive than the requirement for the minimum dwell time in the following two aspects: (i) Unstable modes are allowed to exist, while stability can still be guaranteed by the switching signal; (ii) switches can happen consecutively \cite{zhai2002qualitative}. In this section, we aim to upper bound $\lambda_{\mathrm{max}}(\Gamma_T)$ when the switching signal satisfies an average dwell time constraint. We consider the possible existence of unstable modes: assume $\rho(A_i)<1 $ for $i \in [s_0]$, with some $s_0 < s$, and $\rho(A_j) \geqslant 1$ for all $j\in \{s_0+1,\cdots,s\}$. In addition, there always exist positive real numbers $\lambda_1<1$, $\lambda_2 \geqslant 1$, and $C_v \geqslant 1$ with $v \in [s]$, such that for any positive integer $k$, it holds that $\|A_i^k\| \leqslant C_i \lambda_1^k$ and $\|A_j^k\| \leqslant C_j \lambda_2^k$, where $i \in [s_0]$ and $j \in \{s_0+1,\cdots,s\}$ \cite{zhai2002qualitative}. Then we define
\begin{equation}\label{eq:DefC}
    C = \max_{v\in[s]}\{C_v\}.
\end{equation}



Motivated by the class of switching signals in \cite{zhai2002qualitative}, we consider a slightly different class of switching signals in order to control the growth rate of $\lambda_{\mathrm{max}}(\Gamma_T)$. To introduce it, let $N_{w}(0,t)$ denote the number of switches of $w_t$ within the time step interval $[0,t)$. Let $K^-(0,t)$ and $K^+(0,t)$ denote the number of time steps of the stable and unstable modes within the time step interval $[0,t)$, respectively. Then the considered class of switching signals is defined as follows.

\begin{defn} \label{def:SwitchingAverage} 
Given $\lambda \in (\lambda_1,\lambda_2)$, $\lambda^\star \in (\lambda_1,\lambda]$, $\tau_\mathrm{a}>0$, a non-negative integer $N_0$, and a positive integer $h$, a class of switching signals, denoted by $\mathcal{S}(\tau_\mathrm{a},N_0,\lambda,\lambda^\star,h)$, satisfy the following condition: for any positive integer $j$, it holds that
\begin{enumerate}

    \item 
    \begin{equation} \label{eq:SwitchCons}
 K^-\big((j-1)h,jh\big) \geqslant r K^+\big((j-1)h,jh\big),
    \end{equation}
    where $r\triangleq (\ln \lambda_2 - \ln \lambda^\star)/(\ln \lambda^\star - \ln \lambda_1)$,
and

\item 
\begin{equation} \label{eq:AverageDwellTimeNew}
N_w\big((j-1)h,jh\big) \leqslant \bar{N}_w.
\end{equation}
where  $ \bar{N}_w \triangleq N_0+h/{\tau_\mathrm{a}}$.
\hfill $\blacksquare$ 
\end{enumerate}
\end{defn}
The conditions \eqref{eq:SwitchCons} and \eqref{eq:AverageDwellTimeNew}
constrain the number of unstable modes and switches in the time step intervals with length $h$. The condition \eqref{eq:AverageDwellTimeNew} indicates that if the first $N_0$ switches are ignored, then the average time step interval between two consecutive switches should be at least $\tau_\mathrm{a}$, and $\tau_\mathrm{a}$ is thus called the average dwell time \cite{zhai2002qualitative}. These conditions are extended from the ones in \cite{zhai2002qualitative}, where the same conditions hold with $h=T$ and $j=1$. Intuitively, \eqref{eq:SwitchCons} and \eqref{eq:AverageDwellTimeNew} lead to more evenly distributed unstable modes in a switching sequence, which can facilitate our analysis of $\Gamma_T$. We should also note that \eqref{eq:SwitchCons} and \eqref{eq:AverageDwellTimeNew} are not restrictive if $h$ is sufficiently large.

Given the defined class of switching signals, it is a straightforward extension of \cite[Theorem 3]{zhai2002qualitative} to show the stability of the nominal switched system.
\begin{lem}
\begin{itemize}
    \item If $C=1$ and $\lambda^\star \leqslant 1$ \big(or $\lambda^\star <1$\big) hold, then the switched system is globally marginally stable (or asymptotically stable) for any switching signal $w \in \mathcal{S}(\tau_\mathrm{a},N_0,\lambda,\lambda^\star,h)$ with any $N_0$, $\tau_\mathrm{a}$, $\lambda$ and $h$;
    
    \item If $C>1$, $\lambda \in (\lambda_1,1)$, $\lambda^\star \in (\lambda_1,\lambda)$ hold, and $N_0$ satisfies $N_0 \ln(C) \leqslant -h \ln(\lambda) $\big(or $N_0 \ln(C) < -h \ln(\lambda) $\big), then there exists a $\tau_\mathrm{a}^\star$ such that the switched system is globally marginally stable (or asymptotically stable) for any switching signal $w \in \mathcal{S}(\tau^\star_\mathrm{a},N_0,\lambda,\lambda^\star,h)$.
\end{itemize}
\end{lem}
\begin{proof}
Recall $C$ defined in \eqref{eq:DefC}, and thus \eqref{eq:SwitchCons} implies
\begin{equation}
\|A_{(jh-1:(j-1)h  )}\|_2 \leqslant C^{N_w((j-1)h,jh)} (\lambda^\star)^h, \label{eq:ProofStability1}
\end{equation}
for any positive integer $j$, and thus the first statement holds trivially. If $C>1$, then following a reasoning similar to the proof of \cite[Theorem 3]{zhai2002qualitative}, there exists \begin{equation}
\tau_\mathrm{a}^\star = \ln(C)/(\ln \lambda - \ln \lambda^\star), \label{eq:taustar} \end{equation}
such that given $N_0 \leqslant -h \ln(\lambda)/ \ln(C)$, \eqref{eq:AverageDwellTimeNew} implies
\begin{equation} \label{eq:ProofStability2}
C^{N_w((j-1)h,jh)} (\lambda^\star)^h \leqslant C^{N_0} \lambda^h \leqslant 1. 
\end{equation}
The above last inequality is strict if $N_0 <  -h \ln(\lambda)/ \ln(C) $. 

Then for any $t$, there exists a non-negative integer $b<h$ such that $t = \floor{t/h}h+b$. Based on \eqref{eq:ProofStability1} and \eqref{eq:ProofStability2}, we have $\|x_t\|_2 \leqslant$
\begin{align*} 
&\|A_{(t-1: t-b)}\|_2 \| A_{(\floor{t/h}h-1: (\floor{t/h}-1)h )}\|_2     \cdots  \|A_{(h-1:0)}\|_2\|x_0\|_2        \\
& \leqslant \|x_0\|_2   (C\lambda_2)^{h-1}   \prod_{j=1}^{\floor{t/h}} \Big[ (C)^{N_{w}((j-1)h,jh)} ( \lambda^\star)^{h}  \Big] 
\end{align*}
Therefore, if $N_0 \leqslant -h \ln(\lambda)/ \ln(C)$, we have \eqref{eq:ProofStability2} and thus, the system is marginally stable. If $N_0 < -h \ln(\lambda)/ \ln(C)$, the system is then asymptotically stable. 
\end{proof}
In the above result, the case with $C=1$ covers the situation where all the modes $A_i$ are diagonal matrices; when $C>1$, stability is achieved by upper bounding $N_0$, while in \cite{zhai2002qualitative} $N_0$ can be chosen arbitrarily. However, the upper bound $-h \ln(\lambda)/ \ln(C)$ is not restrictive if $h$ is sufficiently large. 

Finally, with the considered class of switching signals, an upper bound for $\lambda_{\mathrm{max}}(\Gamma_T)$ can be obtained.

\begin{thm}\label{thm:Average}
Given a switched system with a switching signal in $\mathcal{S}(\tau_\mathrm{a},N_0,\lambda,\lambda^\star,h)$, if it satisfies either (i) $C=1$ and $\lambda^\star \leqslant 1$, or (ii) $C>1$, $\lambda \in (\lambda_1,1)$, $\lambda^\star \in (\lambda_1,\lambda)$. $N_0 \ln(C) \leqslant -h \ln(\lambda)$ and $\tau_\mathrm{a} = \tau^\star_\mathrm{a}$ defined in \eqref{eq:taustar}, then it holds that
\begin{equation} \label{eq:UpperBoundAverage}
\lambda_{\mathrm{max}}(\Gamma_T) \leqslant  g(k_0) + g(h)f(k_0)^2 \floor{T/h},
\end{equation}
where the function $g$ is defined in \eqref{eq:FunctionG}, the function $f$ is defined in \eqref{eq:FunctionF}, and $k_0 = T - h\floor{T/h}$.
\end{thm}
\begin{proof}
We first consider the time step interval $[(j-1)h,jh)$, for every positive integer $j$. Let $
\bar{K}^+ \triangleq \floor{h/(1+r)}$
be the maximum allowable number of unstable modes in $[(j-1)h,jh)$ according to \eqref{eq:SwitchCons}. Then for $i=0, 1,\cdots, h$, we have $ \|A_{(jh-1:jh-i)}\|_2 \leqslant $
\begin{align} \label{eq:FunctionF}
f(i) \triangleq \begin{cases}
C^{\bar{N}_w} \lambda_2^i & \text{if $1 \leqslant i \leqslant \bar{K}^+$}  \\
 C^{\bar{N}_w} \lambda_2^{\bar{K}^+} \lambda_1^{i-\bar{K}^+} & \text{if $h>i > \bar{K}^+$}  \\
 1 & \text{if $i = 0$} \\
 a & \text{if $i = h$} 
 \end{cases},
\end{align}
where $a= (\lambda^\star)^h$ in case (i) based on \eqref{eq:ProofStability1}, or $a = C^{N_0} \lambda^h$ in case (ii) due to \eqref{eq:ProofStability2}. Then it holds that $\sum_{i=0}^{b-1}  \|A_{(jh-1:jh-i)}\|_2^2 \leqslant  g(b) \triangleq$
\begin{align}
1+  C^{2\bar{N}_w} \Big[\sum_{j=1}^{\mathrm{min}\{b-1, \bar{K}^+\}} \lambda_2^{2j} +  \lambda_2^{2 \bar{K}^+} \sum_{k=\bar{K}^+ +1}^{\mathrm{min}\{b-1,h-1\}} \lambda_1^{2(k- \bar{K}^+  )} \Big], \label{eq:FunctionG}
\end{align}
where $b \in  \{1,2,\dots,h\}$, and $g(0) \triangleq 0$. Then for any $T$, there exists a $k_0<h$ such that $T = \floor{T/h}h+k_0$. If $i > k_0$, based on \eqref{eq:ProofStability2} we have
$\|A_{(T-1:T-i)}\|_2 \leqslant 
a^{\floor{(i-k_0)/h}} \|A_{(T-1:T-k_0)}\|_2  \|A_{(jh-1:jh-l)}\|_2, $
for some positive integer $j$ and some $l \in \{0,\dots,h-1\}$. Therefore, it holds  $\lambda_{\mathrm{max}}(\Gamma_T) \leqslant$
\begin{align}
&\sum_{i=0}^{T-1} \|A_{(T-1:T-i)}\|_2^2 \leqslant  \sum_{j=0}^{k_0-1}\|A_{(T-1:T-j)}\|_2^2 \nonumber \\
& + \|A_{(T-1:T-k_0)}\|_2^2 \Big(  \sum_{j=0}^{\floor{T/h}-1} a^{2j} \sum_{l=0}^{h-1} \|A_{(k_j h-1:k_j h-l)}\|_2^2   \Big) \nonumber \\
& \leqslant g(k_0) +g(h) f(k_0)^2 \sum_{j=0}^{\floor{T/h}-1} a^{2j}, \label{eq:FinalBoundAverage}
\end{align}
where $k_j=1,\cdots,\floor{T/h}$. The result is obtained from $a= (\lambda^\star)^h \leqslant 1$ in case (i) and $a = C^{N_0} \lambda^h \leqslant 1$ in case (ii).
\end{proof}
The bound of the above result shows $\lambda_{\mathrm{max}}(\Gamma_T) = \mathcal{O}(T)$ for marginally stable systems. When the system is asymptotically stable, a bound can be obtained by exploiting \eqref{eq:FinalBoundAverage} and $a<1$. To better interpret the upper bound, we consider the special case where $h$ is a factor of $T$.
\begin{coro}
In the setting of Theorem~\ref{thm:Average}, if $T-h\floor{T/h}=0$ also satisfies, then it holds that $\lambda_{\mathrm{max}}(\Gamma_T) \leqslant \floor{T/h} \times$
\begin{align*}
 \bigg[&  1 +   C^{2\bar{N}_w} \bigg( \frac{ \lambda_2^{2(\bar{K}^++1)} -\lambda_2^2}{\lambda_2^2 - 1 } 
+\lambda_2^{2 \bar{K}^+} \frac{\lambda_1^2- \lambda_1^{2(h- \bar{K}^+)} }{1- \lambda_1^2} \bigg) \bigg] ,
\end{align*}
where $\bar{K}^+ \triangleq \floor{h/(1+r)}$. \hfill $\blacksquare$
\end{coro}
Given $h$, the above bound increases if more switches are allowed, i.e., a larger $\bar{N}_w$, or if more unstable modes can be active, i.e., a larger $\bar{K}^+$. In addition, the bound admits an exponential growth rate in $h$: Let $C=1$, then $\lambda_{\mathrm{max}}(\Gamma_T)$ is $\mathcal{O}(\floor{T/h} \lambda_2^{k_1 h}) $ for some positive constant $k_1$, and if $T \geqslant h$ holds, this bound increases exponentially as $h$ increases. This exponential increase is due to that a larger time interval $h$ in Definition~\ref{def:SwitchingAverage} allows the unstable modes to be active continuously for a longer period, and it further leads to a linear dependence on $\sqrt{h}$ in the estimation error bound \eqref{eq:PreBound2}, in contrast to the logarithmic dependence on the switching parameters in Corollaries~\ref{coro:ArbitrarySwitch} and \ref{coro:MinimumDwell}. Therefore, $h$ should be limited to avoid the potential significant increase of the estimation error.

With the above bounds, a bound for the LS estimation error can be obtained by combining \eqref{eq:UpperBoundAverage} and \eqref{eq:PreBound2}. A numerical example is shown in Fig.~\ref{fig:exam0}, and the data is generated by a two-mode system which contains $A_1=\mathrm{diag}(0.5,0.5)$, $A_2=\mathrm{diag}(2,2)$ and satisfies the case (i) in Theorem~\ref{thm:Average}.
\begin{figure}[h]
\begin{minipage}{0.5\textwidth}
\hspace*{-0.7cm}
\centering
\includegraphics[scale=0.19]{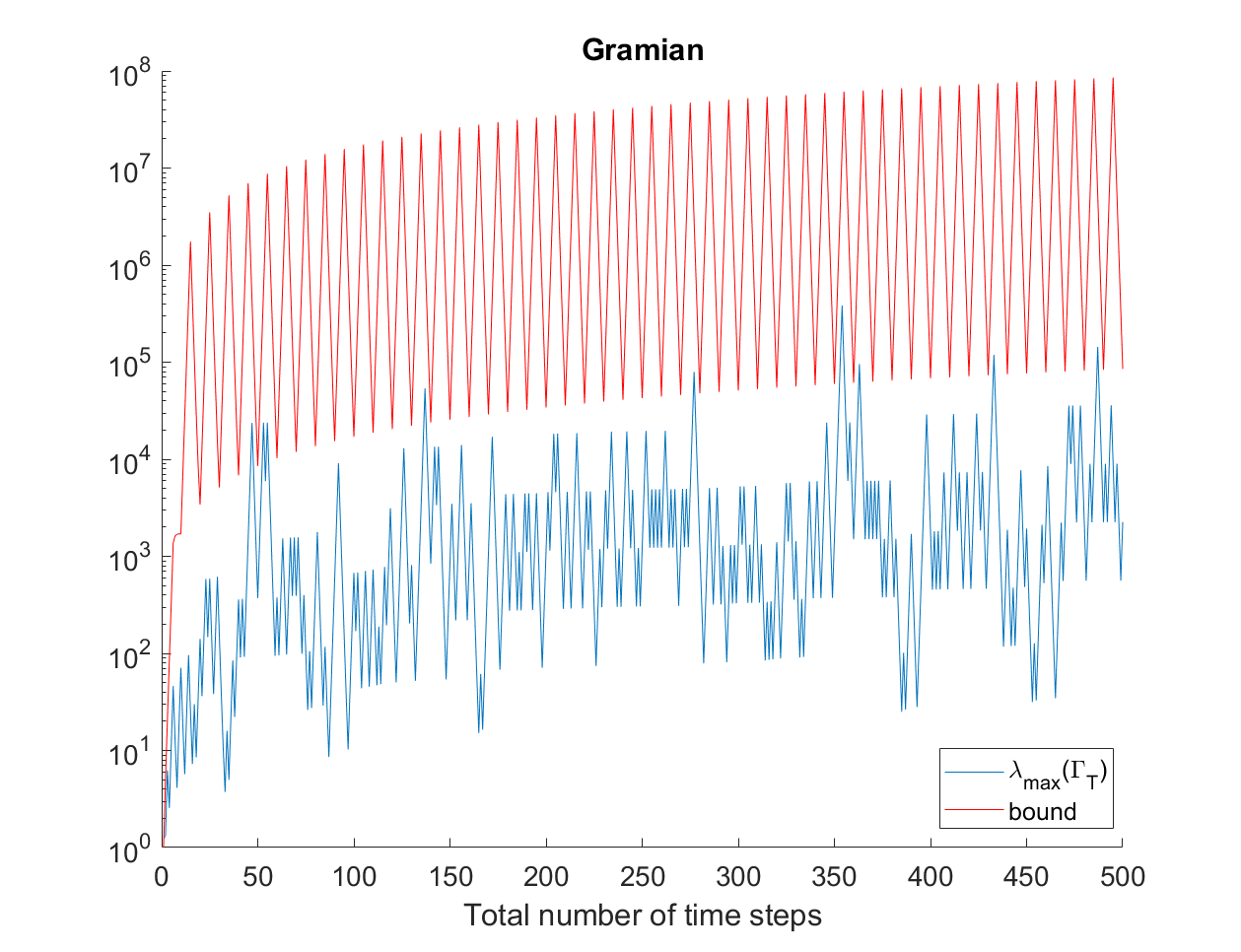}
\\(a)
\end{minipage}
\begin{minipage}{0.5\textwidth}
\hspace*{-0.3cm}
\centering
\includegraphics[scale=0.2]{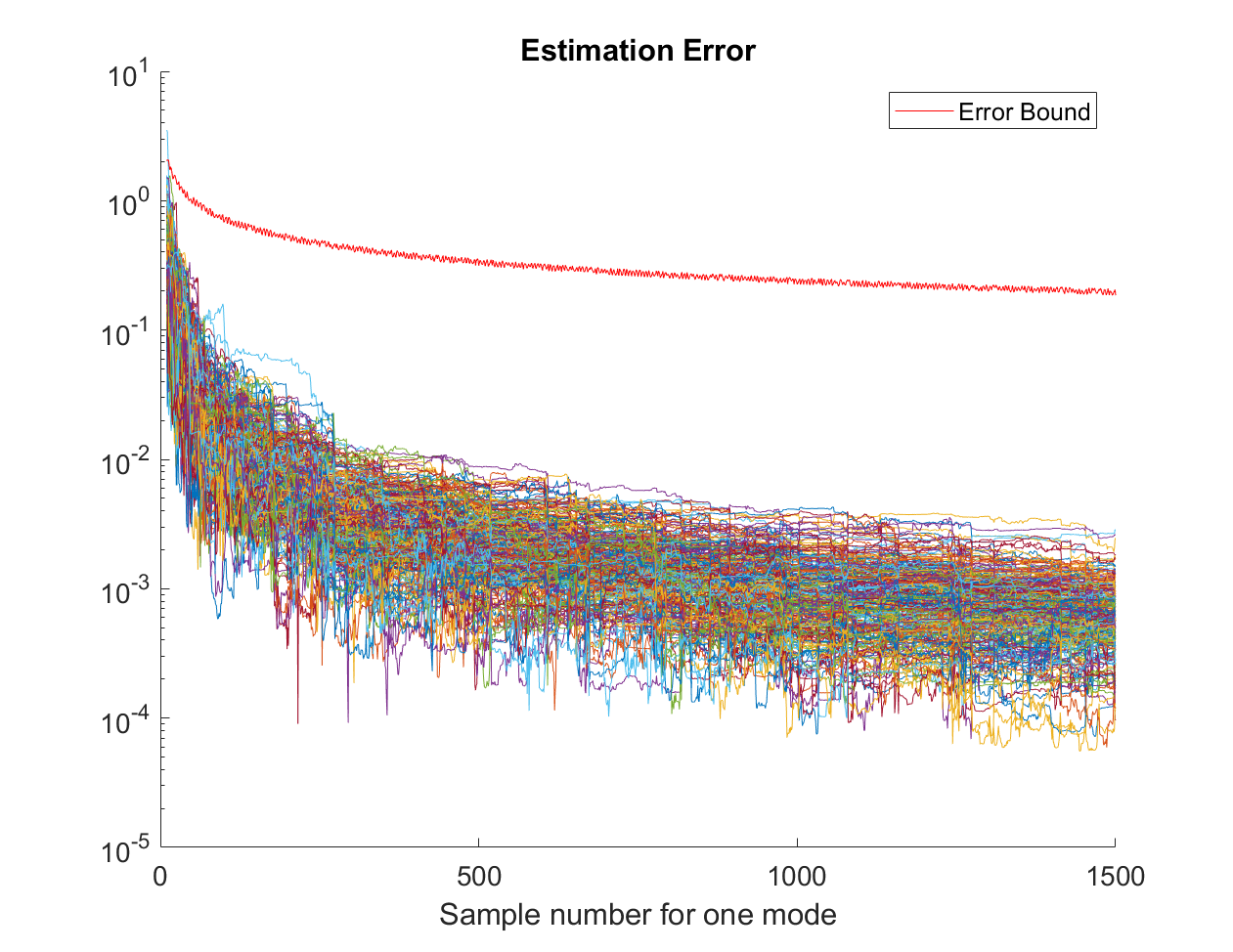}
\\(b)
\end{minipage}
\caption{In (a), the bound \eqref{eq:UpperBoundAverage} is compared with $\lambda_{\mathrm{max}}(\Gamma_T)$, which shows that the bound correctly capture the behavior of $\lambda_{\mathrm{max}}(\Gamma_T)$; the resulting estimation error bound for the unstable mode is shown in (b) and is compared to $300$ error trajectories resulting from $300$ noise realizations. The bound captures the behavior of the error up to a constant, i.e., its magnitude is conservative but captures the decay trend of the error.}
\label{fig:exam0}
\end{figure}

\section{CONCLUSIONS}
Finite-sample error bounds are developed for the LS estimation of switched systems, such that the bounds capture the effect of the parameters of the switching strategies. It is shown that when there are only stable modes, the bound is logarithmic of the switching parameters; however, the presence of unstable modes leads to a linear increase of the error bound as the change of the switching parameter. This suggests that when there are unstable modes, the switching signal should be properly designed to avoid a significant increase in the estimation error. While the developed theoretical error bounds are conservative as they concern the worst-case estimation error under the considered classes of switching signals, they reveal how the estimation error scales with the sample size and the important parameters of the switched systems. Future work includes the application of the developed bounds to analyze the sample complexity of hybrid controllers \cite{matni2019self}, the development of less conservative bounds for practical applications, and the consideration of output measurements and unmeasured switching signal. 


\addtolength{\textheight}{-12cm}  



\bibliographystyle{IEEEtran} 
\bibliography{LibraryFiniteSample.bib}

\begin{thebibliography}{10}
\providecommand{\url}[1]{#1}
\csname url@samestyle\endcsname
\providecommand{\newblock}{\relax}
\providecommand{\bibinfo}[2]{#2}
\providecommand{\BIBentrySTDinterwordspacing}{\spaceskip=0pt\relax}
\providecommand{\BIBentryALTinterwordstretchfactor}{4}
\providecommand{\BIBentryALTinterwordspacing}{\spaceskip=\fontdimen2\font plus
\BIBentryALTinterwordstretchfactor\fontdimen3\font minus
  \fontdimen4\font\relax}
\providecommand{\BIBforeignlanguage}[2]{{%
\expandafter\ifx\csname l@#1\endcsname\relax
\typeout{** WARNING: IEEEtran.bst: No hyphenation pattern has been}%
\typeout{** loaded for the language `#1'. Using the pattern for}%
\typeout{** the default language instead.}%
\else
\language=\csname l@#1\endcsname
\fi
#2}}
\providecommand{\BIBdecl}{\relax}
\BIBdecl

\bibitem{matni2019self}
N.~Matni, A.~Proutiere, A.~Rantzer, and S.~Tu, ``From self-tuning regulators to
  reinforcement learning and back again,'' in \emph{2019 IEEE 58th Conference
  on Decision and Control (CDC)}, 2019, pp. 3724--3740.

\bibitem{simchowitz2018learning}
M.~Simchowitz, H.~Mania, S.~Tu, M.~I. Jordan, and B.~Recht, ``Learning without
  mixing: {T}owards a sharp analysis of linear system identification,'' in
  \emph{Conference on Learning Theory}.\hskip 1em plus 0.5em minus 0.4em\relax
  PMLR, 2018, pp. 439--473.

\bibitem{sarkar2019near}
T.~Sarkar and A.~Rakhlin, ``Near optimal finite time identification of
  arbitrary linear dynamical systems,'' in \emph{International Conference on
  Machine Learning}.\hskip 1em plus 0.5em minus 0.4em\relax PMLR, 2019, pp.
  5610--5618.

\bibitem{jedra2020finite}
Y.~Jedra and A.~Proutiere, ``Finite-time identification of stable linear
  systems optimality of the least-squares estimator,'' in \emph{2020 59th IEEE
  Conference on Decision and Control (CDC)}, 2020, pp. 996--1001.

\bibitem{djehiche2021non}
B.~Djehiche and O.~Mazhar, ``Non asymptotic estimation lower bounds for {LTI}
  state space models with {C}ramer-{R}ao and van {T}rees,'' \emph{arXiv
  preprint arXiv:2109.08582}, 2021.

\bibitem{sarkar2021finite}
T.~Sarkar, A.~Rakhlin, and M.~A. Dahleh, ``Finite time {LTI} system
  identification.'' \emph{Journal of Machine Learning Research}, vol.~22, pp.
  1--61, 2021.

\bibitem{oymak2021revisiting}
S.~Oymak and N.~Ozay, ``Revisiting {H}o-{K}alman based system identification:
  robustness and finite-sample analysis,'' \emph{IEEE Transactions on Automatic
  Control}, 2021, {E}arly Access.

\bibitem{lauer2019hybrid}
F.~Lauer and G.~Bloch, ``Hybrid system identification,'' in \emph{Hybrid System
  Identification}.\hskip 1em plus 0.5em minus 0.4em\relax Springer, 2019, pp.
  77--101.

\bibitem{zhang2016switched}
Q.~Zhang, Q.~Wang, and G.~Li, ``Switched system identification based on the
  constrained multi-objective optimization problem with application to the
  servo turntable,'' \emph{International Journal of Control, Automation and
  Systems}, vol.~14, no.~5, pp. 1153--1159, 2016.

\bibitem{chen2014application}
S.~Chen, L.~Jiang, W.~Yao, and Q.~H. Wu, ``Application of switched system
  theory in power system stability,'' in \emph{2014 49th International
  Universities Power Engineering Conference (UPEC)}.\hskip 1em plus 0.5em minus
  0.4em\relax IEEE, 2014, pp. 1--6.

\bibitem{sarkar2019nonparametric}
T.~Sarkar, A.~Rakhlin, and M.~Dahleh, ``Nonparametric system identification of
  stochastic switched linear systems,'' in \emph{2019 IEEE 58th Conference on
  Decision and Control (CDC)}, 2019, pp. 3623--3628.

\bibitem{sattar2021identification}
Y.~Sattar, Z.~Du, D.~A. Tarzanagh, L.~Balzano, N.~Ozay, and S.~Oymak,
  ``Identification and adaptive control of {M}arkov jump systems: {S}ample
  complexity and regret bounds,'' \emph{arXiv preprint arXiv:2111.07018}, 2021.

\bibitem{sayedana2022switch}
B.~Sayedana, M.~Afshari, P.~E. Caines, and A.~Mahajan, ``Consistency and rate
  of convergence of switched least squares system identification for autonomous
  switched linear systems,'' \emph{Proceedings of Machine Learning Research},
  vol. 144, pp. 1--17, 2022.

\bibitem{lin2009stability}
H.~Lin and P.~J. Antsaklis, ``Stability and stabilizability of switched linear
  systems: a survey of recent results,'' \emph{IEEE Transactions on Automatic
  control}, vol.~54, no.~2, pp. 308--322, 2009.

\bibitem{morris2021marginally}
I.~D. Morris, ``Marginally unstable discrete-time linear switched systems with
  highly irregular trajectory growth,'' \emph{Systems \& Control Letters}, vol.
  163, p. 105216, 2022.

\bibitem{vidal2008recursive}
R.~Vidal, ``Recursive identification of switched {ARX} systems,''
  \emph{Automatica}, vol.~44, no.~9, pp. 2274--2287, 2008.

\bibitem{vershynin2010introduction}
R.~Vershynin, ``Introduction to the non-asymptotic analysis of random
  matrices,'' \emph{arXiv preprint arXiv:1011.3027}, 2010.

\bibitem{zhai2002qualitative}
G.~Zhai, B.~Hu, K.~Yasuda, and A.~N. Michel, ``Qualitative analysis of
  discrete-time switched systems,'' in \emph{Proceedings of the 2002 American
  Control Conference}, vol.~3, 2002, pp. 1880--1885.

\end{thebibliography}

\end{document}